\documentclass[letterpaper,english,12pt]{article}

\usepackage{pgfplots}
\pgfplotsset{compat=1.14}
\usepackage[T1]{fontenc}
\usepackage[utf8]{inputenc}
\usepackage{amsmath,amssymb,amsthm}
\usepackage{libertine}
\usepackage{inconsolata}
\usepackage{amsfonts}

\usepackage{fullpage}
\usepackage{amsmath}
\usepackage{graphicx}
\usepackage{siunitx}
\usepackage{hyperref}
\hypersetup{
  citecolor=blue,
  colorlinks=true,
}

\usepackage{cleveref}
\usepackage{booktabs}
\usepackage{xcolor}
\usepackage{braket}
\usepackage{amsfonts}
\usepackage{amssymb}
\usepackage{mathabx}
\usepackage{systeme}
\usepackage{mdframed}
\usepackage{braket}
\usepackage{mathtools}
\usepackage{bm}
\usepackage{varwidth}
\usepackage{amsthm}
\usepackage{subfig}
\usepackage[export]{adjustbox}
\usepackage{algorithm}
\usepackage{algorithmic}
\usepackage{changepage}

\usepackage{pgf}
\usepackage{tikz}
\usepackage{verbatim}
\usepackage{authblk}

\usetikzlibrary{arrows,automata}
\usetikzlibrary{positioning}

\newtheorem{theorem}{Theorem}

\DeclarePairedDelimiter\ceil{\lceil}{\rceil}

\DeclareUnicodeCharacter{00A0}{ }

\usepackage{graphicx}
\makeatletter
\newcommand*\bigcdot{\mathpalette\bigcdot@{.5}}
\newcommand*\bigcdot@[2]{\mathbin{\vcenter{\hbox{\scalebox{#2}{$\m@th#1\bullet$}}}}}

\newcommand{\llbracket}{{[\![}}
\newcommand{\rrbracket}{{]\!]}}

\makeatother

\title{Synthesizing quantum circuits via numerical optimization}
\author[1,3]{Timothée Goubault de Brugière}
\author[1]{Marc Baboulin}
\author[2]{Benoît Valiron}
\author[3]{Cyril~Allouche}

\date{}

\affil[1]{Université Paris-Saclay, CNRS,                  Laboratoire de recherche en informatique, 91405, Orsay, France}
\affil[2]{Université Paris-Saclay, CNRS, CentraleSupélec, Laboratoire de Recherche en Informatique, 91405, Orsay, France}
\affil[3]{Atos Quantum Lab, Les Clayes-sous-Bois, France}

\begin{document}

\maketitle

\begin{abstract}
We provide a simple framework for the synthesis of quantum circuits based on a numerical optimization algorithm. This algorithm is used in the context of the trapped-ions technology. We derive theoretical lower bounds for the number of quantum gates required to implement any quantum algorithm. Then we present numerical experiments with random quantum operators where we compute the optimal parameters of the circuits and we illustrate the correctness of the theoretical lower bounds. We finally discuss the scalability of the method with the number of qubits.
\end{abstract}

\section{Introduction}

Quantum computing, introduced and theorized about 30 years ago, is
still in its infancy: current technological devices can only handle a
few qubits. This new paradigm of computation however shows great
promises, with potential applications ranging from high-performance
computing~\cite{harrow2009quantum} to machine learning and big data~\cite{kerenidis2017recommendation}.
Quantum algorithms are usually described via quantum circuits, i.e., series
of elementary operations in line with the technological specificity of the
hardware. The mathematical formalism for quantum computation is the
theory of (finite dimensional) Hilbert spaces: a quantum circuit is
represented as a unitary operator \cite{nielsen2011quantum},
independently from the machine support on which the algorithm will be
executed. Establishing a link between the unitary operators described as
matrices and the unitary operators described as circuits is therefore
essential, if only to better understand how to design new
algorithms. Obtaining the matrix from the circuit can be done either
by running the circuit on a quantum hardware (plus some tomography) or
via a simulation on a classical computer
\cite{bravyi2018simulation,PhysRevLett.91.147902}. Obtaining the
circuit from the matrix is more complicated and fits into a more
general problematic called {\it quantum compilation} i.e., how to
translate a quantum operator described in an unknown form for the
targeted hardware into a sequence of elementary instructions
understandable by the machine. Therefore converting a matrix into a
circuit is in all circumstances a compilation step because no hardware
can directly accept a matrix as input: this particular task of
compilation is called {\it quantum circuit synthesis} and is the
central concern of this work.
This process must be as automated and optimized as possible in terms of conventional or quantum resources. This task is inevitably hard in the general case: the classical memory required to store the unitary matrix is exponential in the number of qubits and it has been shown that an exponential number of quantum gates is necessary to implement almost every operator~\cite{knill1995approximation}.
As a result both classical and quantum complexity follow an
exponential growth and any methods will be quickly limited by the size
of the problem. Yet being able to synthesize quickly and/or optimally
any quantum operator or quantum state on a few qubits is a crucial
challenge that could be crucial for some applications. In the NISQC (Noisy Intermediate Scale Quantum Computer)
era \cite{NISQC}, compressing circuits as much as possible will be crucial. And in
general, such a procedure can be
integrated as a subtask in a peep-hole optimization framework.

\paragraph{Our contribution.}
 The quantum circuits addressed in this paper correspond to the specific technology of trapped-ions quantum computer which requires using specific quantum gates. This technology holds great promise to cross the passage to scale: quantum circuits with trapped-ions technology have great fidelities and the qubits have a long decoherence time. In other words the noise in the results is low and long time computations can be performed before the system interferes with the environment and loses its information. Moreover the particular architecture of trapped-ions quantum circuits makes the problem simpler for numerical optimization because as we will see it only involves the optimization of continuous real parameters. We provide a simple optimization framework and use it to synthesize generic quantum operators. More specifically, we derive a lower bound on the number of entangling gates necessary to implement any quantum algorithm and propose numerical experiments that confirm the correctness of this bound.

\medskip
\noindent
This preprint has been submitted to ICCS 2019.

 \paragraph{Plan of the paper.}
 In Section \ref{background}, we give the main definitions related to
 quantum circuits and quantum gates and we summarize the state of the art in quantum circuit synthesis.
 Then in Section \ref{lower} we develop in more details the modeling of a circuit as a parameterized topology and the question of the minimum-sized topology necessary to implement any operator. In Section \ref{framework} we introduce a generic optimization framework formalizing the circuit synthesis as a classical optimization problem. In Section \ref{experiments} we apply this framework to optimally synthesize generic unitary matrices with the trapped-ions natural set of gates. We also discuss the scalability of this method and how we can naturally trade the optimality of the final quantum circuit for shorter computational time. We conclude in Section \ref{conclusion}.

 \paragraph{Notations.}
 Throughout this paper we will use the following
 notations. $\mathcal{U}(n)$ denotes the set of unitary matrices of
 size $n$, i.e.
 $\mathcal{U}(n) = \{ M \in \mathbb{C}^{n \times n} \; | \;
 M^{\dag}M = I \}$, where $I$ is the identity matrix and $M^{\dag}$ is
 the conjugate transpose of the matrix $M$. The special unitary group $\mathcal{SU}(n)$ is the group of $n \times n$ unitary matrices with determinant 1.
 $\|x\|=\sqrt{x^{\dag}x}$
 refers to the Euclidean norm of a vector $x \in \mathbb{C}^{n}$ and
 $A \otimes B$ denotes the Kronecker product~\cite{KRONECKER} of two
 matrices $A$ and $B$.

\section{Background and state of the art}\label{background}

\subsection{Main notions in quantum circuits}
The basic unit of information in quantum information is the quantum bit and is formalized as a complex linear superposition of two basis states, usually called the state '$0$' and the state '$1$'. More generally when manipulating a set of $n$ qubits we modify a quantum state 
\[ \ket{\psi} = \sum_{\mathbf{b} \in \{0,1\}^n} \alpha_{\mathbf{b}} \ket{\mathbf{b}} \] 
which is a normalized complex linear superposition of all the possible n-bitstring values. So a quantum state on $n$ qubits can be written as a unit vector in $\mathbb{C}^{2^n}$ and by the laws of quantum physics any operation on this quantum state can be represented as a left-multiplication of its vector representation by a unitary matrix $U \in \mathcal{U}(2^n)$, if we except the measurement operation.

Quantum algorithms are described via quantum circuits---the quantum
analog of logical circuits. An example is given in Figure
\ref{circuit1}. Each horizontal wire corresponds to one qubit and
quantum gates --represented as boxes, or vertical apparatus---are
sequentially applied from left to right to different subsets of qubits
resulting in a more complex unitary operator. Matrix multiplication
$BA$ of two operators $A$ and $B$ enables to apply sequentially the
operator $A$ then $B$. To compose two operators $A$ and $B$ acting on
different qubits we use the Kronecker product $\otimes$.

\begin{figure}
\newcommand{\gate}[4]{\draw[fill=white] (#1-#3,#2-#3) rectangle (#1+#3,#2+#3); \node at (#1,#2) {#4} }
\centering
\begin{tikzpicture}[x=3ex,y=-3ex]
 \draw (0,1) -- (12,1) ;
 \draw (0,2.5) -- (12,2.5) ;
 \draw (0,4) -- (12,4) ;
 \gate{1}{1}{0.6}{$R_z$} ;
 \draw[fill=black] (3,1) circle (0.1) ;
 \draw (3,1) -- (3,2.8) ;
 \draw (3,2.5) circle (0.3) ; 
 \gate{5}{4}{0.6}{$R_x$} ;
 \gate{7}{4}{0.6}{$R_y$} ;
 \gate{9}{4}{0.6}{$R_z$} ;
 \draw[fill=black] (11,4) circle (0.1) ;
 \draw (11,4) -- (11,2.2) ;
 \draw (11,2.5) circle (0.3) ;
\end{tikzpicture}
\caption{Example of quantum circuit}
\label{circuit1}
\end{figure}
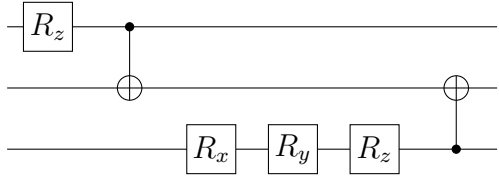

For a given technology not all quantum gates are realizable and only a
few of them are directly implementable: the so-called elementary
gates. For example, we have the Hadamard gate
$H = \frac{1}{\sqrt{2}} \begin{pmatrix} 1 & 1 \\ 1 &
  -1 \end{pmatrix}$, or the one-qubit rotations of angle
$\theta \in \mathbb{R}$ along the x, y and z axis respectively defined
by
$R_x(\theta) = \begin{pmatrix} \cos(\theta) & -i\sin(\theta) \\
  -i\sin(\theta) & \cos(\theta) \end{pmatrix} , R_y(\theta)
= \begin{pmatrix} \cos(\theta) & -\sin(\theta) \\ \sin(\theta) &
  \cos(\theta) \end{pmatrix}, R_z(\theta) = \begin{pmatrix}
  e^{-i\theta/2} & 0 \\ 0 & e^{i\theta/2} \end{pmatrix}$. Fortunately
some subsets of elementary gates have been shown to be ``universal'',
which means that any quantum operator can be implemented by a quantum
circuit containing only gates from this set. Depending on the
technology used, the universal sets available will be different. With
superconducting qubits or linear optics, the natural set of gates is
the set of one-qubit gates $\mathcal{SU}(2)$, also called local gates
because they act on one qubit only, combined with the entangling CNOT
(Controlled-NOT) gate which is a NOT gate controlled by one qubit. For
instance, in the 2-qubits case, the CNOT gate controlled by the first
qubit and applied to the second one can be represented by the matrix
$\begin{pmatrix} 1 & 0 & 0 & 0 \\ 0 & 1 & 0 & 0 \\ 0 & 0 & 0 & 1 \\ 0
  & 0 & 1 & 0
\end{pmatrix}$.

In this paper we focus on the technology of trapped ions which uses a different universal set of gates, the available gates are:  
\begin{itemize}
\item local $R_z$ gates, 
\item global $R_x$ gates i.e local $R_x$ are applied to every qubit with the same angle,
\item the entangling Mølmer–Sørensen gate (MS gate) defined by 
\[ MS(\theta) = e^{-i \theta (\sum_{i=1}^n \sigma^i_x)^2/4}. \]
where $\sigma^i_x$ is the operator $X$ applied to the $i$-th qubit. 
\end{itemize}

Any quantum operator has an abstract representation using a unitary
matrix and it is essential to be able to switch from a quantum
operator given by a quantum circuit to a quantum operator given by its
unitary matrix and vice-versa. From a quantum circuit to a unitary
matrix this is the problem of the {\it simulation of a quantum
  circuit}. Finding a quantum circuit implementing a unitary matrix is
the problem of the {\it synthesis of a quantum circuit} which is the
central concern of this paper. We distinguish two different synthesis
problems: the first one consists in implementing a complete unitary
matrix and the second one consists in preparing a specific quantum
state as output of the circuit applied to the state $\ket{000...00}$:
this is the {\it state preparation problem}.

\subsection{State of the art in quantum circuit synthesis}
Most quantum algorithms are still designed by hand \cite{Grover:1996:FQM:237814.237866,doi:10.1137/S0036144598347011} even though some circuits have been found via automatic processes \cite{li2016approximate}. For the automatic synthesis of quantum circuits there are mainly algebraic methods: we can mention the Quantum Shannon Decomposition (QSD) that gives the lowest number of gates in the general case \cite{1629135} or the use of Householder matrices to achieve the synthesis in a much lower time~\cite{Householder}.
 The (H,T) framework is used for one-qubit synthesis \cite{kliuchnikov2013fast} although it can be extended to the multi-qubits case \cite{PhysRevA.87.032332}. For the particular case of state preparation - the synthesis of one column of the quantum operator -  we have devices using multiplexors \cite{mottonen2005decompositions} or Schmidt decomposition \cite{PhysRevA.83.032302}. 

 Although the asymptotic complexity of the methods has significantly
 decreased with the years, some progress can still be made. The
 motivation behind this article is to use well-known numerical
 optimization methods in hope of reducing at its maximum the final
 quantum resources. Using heuristics or classical optimization methods
 to synthesize circuits is not new. The BFGS (Broyden-Fletcher-Goldfarb-Shanno) algorithm
 \cite{wright1999numerical} has already been used to synthesize
 trapped-ions circuits~\cite{martinez2016compiling} and machine
 learning techniques have been used in the case of photonic computers
 \cite{arrazola2018machine}. Genetic algorithms have also been used in
 a general context \cite{lukac2003evolutionary} or for the specific
 IBM quantum computer \cite{li2016approximate}. However, these works are purely
   experimental: the overall optimality of their solution is not
   discussed.

 We tackle the problem of the optimality of
   the solution
 by building on the work in \cite{shende2004minimal} that provides a theoretical lower bound of the number of entangling gates necessary in the quantum circuit synthesis problem. The idea is to count the number of degrees of freedom in a quantum circuit and show that this number has to exceed a certain threshold to be sure that an exact synthesis is possible for any operator. To our knowledge numerical methods have not been used in order to address the issue of the lower bound and the more general issue of the minimum quantum resources that are necessary to synthesize a quantum circuit.

\section{Lower bounds for the synthesis of trapped-ions quantum circuits}\label{lower}

The problem of computing the minimal number of quantum gates necessary
to implement an operator remains
open. Knill~\cite{knill1995approximation} showed that for the entire
set of quantum operators (up to a zero measure set) we need an
exponential number of gates and using a polynomial number of gates is
as efficient as using a random circuit in the general case. The
special case of circuits built from $\{\mathcal{SU}(2), CNOT \}$ has
been analyzed in~\cite{shende2004minimal}, where quantum circuits are
modeled as instantiations of circuit topologies consisting of constant
gates and parameterized gates. The unspecified parameters are the
degrees of freedom (DOF) of the topology. For instance the circuit
given in Figure \ref{circuit1} can be considered as a topology with at
most 4 degrees of freedom (one for each rotation) and giving precise
angles to the rotations is an instantiation of the topology.  As a
consequence a topology with $k$ degrees of freedom can be represented by a smooth function
\begin{align} f: \mathbb{R}^k \to \mathcal{U}(2^n) \end{align} 
that maps the values of angles to the space of unitary matrices of size $2^n$.

We are interested in the image of the function $f$. 
If a topology $f$ on $n$ qubits can implement any n-qubits operator, i.e., for any operator $U$ on $n$ qubits there exists a vector of angles $x$ such that $f(x) = U$, then we say that the topology is universal. Now what is the minimum number of gates necessary to obtain a universal topology ? The best lower bound for this minimum in the case of $\{CNOT, \mathcal{SU}(2)\}$ circuits is given in \cite{shende2004minimal}. In this section, we derive a lower bound in the case of trapped-ions circuits using a similar reasoning.

\begin{theorem}
A topology composed of one-qubit rotations and parameterized MS gates cannot be universal with fewer than $\ceil[\big]{\frac{4^n - 3n - 1}{2n+1}}$ MS gates.
\end{theorem}

\begin{proof}

First we use Sard's theorem~\cite{guillemin2010differential} to claim that the image of $f$ is of measure 0 if $k = \#DOF < dim(\mathcal{U}(2^n))$. Hence to be sure that we can potentially cover the whole unitary space we need

\begin{align} \#DOF \geq dim(\mathcal{U}(2^n)) = 4^n \label{sard} \end{align}

Next we give a normal form to trapped-ion circuits in order to count the number of DOF. MS gates operate on all qubits, they are diagonal in the basis $H^{\otimes n} = \bigotimes_{i=1}^{n} H$ obtained by applying an Hadamard gate to each qubit, the so-called "$\ket{+}/\ket{-}$" basis. We have 
\begin{equation} MS(\theta) = H^{\otimes n} \times D(\theta) \times H^{\otimes n} \label{MS_diag} \end{equation}
with 
\begin{equation} D(\theta) = \text{diag}([e^{(n - \textit{Hamm}(i))^2 \times \theta}]_{i=0..2^n-1}) \label{diag} \end{equation}
where $\textit{Hamm}$ is the Hamming weight in the binary alphabet. 

First we can merge the Hadamard gates appearing in Equation (\ref{MS_diag}) with the local gates so that we can consider that our circuits are only composed of local gates and diagonal gates given by Equation~(\ref{diag}). Then  we can write each local gate $U$ as 
\begin{equation} U = R_z(\alpha) \times R_x(-\pi/2) \times R_z(\beta) \times R_x(\pi/2) \times R_z(\gamma) \label{one-qubit-decomposition} \end{equation}
where $\alpha, \beta, \gamma$ parameterize the unitary matrix $U$. Because the MS gates are now diagonal we can commute the first $R_z$ so that it merges with the next local unitary matrices. By doing this until we reach the end of the circuit we finally get a quantum circuit for trapped-ions hardware with the following basic subcircuit: 
\begin{itemize}
	\item a layer of local $R_z$ gates,
	\item a layer of global $R_x(\pi/2)$ gates,
	\item a layer of local $R_z$ gates,
	\item a layer of global $R_x(-\pi/2)$ gates,
	\item an MS gate (given in its diagonal form)
\end{itemize}

\begin{figure*}
  \newcommand{\gate}[4]{\draw[fill=white] (#1-#3,#2-#3) rectangle (#1+#3,#2+#3); \node at (#1,#2) {\scalebox{.8}{#4}} } 
  \newcommand{\gatel}[5]{\draw[fill=white] (#1-#4,#2-#3) rectangle (#1+#4,#2+#3); \node at (#1,#2) {\scalebox{.8}{#5}} }
  \newcommand{\gatesl}[2]{\gate{#1}{#2}{.6}{$R_z$} ;
    \gatel{#1+2}{#2}{.6}{1}{$R_x\!(\frac{\pi}2)$} ;
    \gate{#1+4}{#2}{.6}{$R_z$} ;
    \gatel{#1+6}{#2}{.6}{1}{$R_x\!(\frac{-\pi}2)$}}
  \newcommand{\gatesarray}[2]{\gatesl{#1}{#2};\gatesl{#1}{#2+1.5};\gatesl{#1}{#2+3}}
  \newcommand{\msbox}[2]{\draw[fill=white] (#1-.9,#2-2.1) rectangle (#1+.9,#2+2.1);
    \node at (#1,#2) {\scalebox{.8}{$\mathit{MS}$}}}
  \centering
  \begin{tikzpicture}[x=3ex,y=-3ex]
    \draw (0,1) -- (30,1) ;
    \draw (0,2.5) -- (30,2.5) ;
    \draw (0,4) -- (30,4) ;
    \gatesarray{1}{1};
    \gatesarray{11}{1};
    \gatesarray{21}{1};
    \gate{29}{1}{.6}{$R_z$} ;
    \gate{29}{2.5}{.6}{$R_z$} ;
    \gate{29}{4}{.6}{$R_z$} ;
    \msbox{9.2}{2.5} ;
    \msbox{19.2}{2.5} ;
  \end{tikzpicture}
  \caption{Generic quantum circuit on 3 qubits for the trapped-ions technology}
\label{MS_circuit}
\end{figure*}
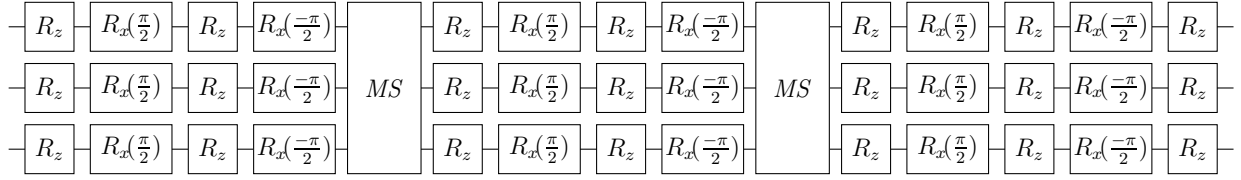

This subcircuit is repeated $k$ times for a circuit with $k$ MS gates. Ultimately, the circuit ends with a layer of rotations following the decomposition (\ref{one-qubit-decomposition}) on each qubit. An example of a quantum circuit on 3 qubits with 2 MS gates is given in Figure \ref{MS_circuit}. The angles of the $R_z$ rotations are omitted for clarity. The only parameters of such generic circuits are: 
\begin{itemize}
	\item the angles of the $R_z$ rotations,
	\item the number and the angles of the MS gates.
\end{itemize}
Each elementary rotation (around the $x,y,z$ axis) is parameterized by one angle so it can only bring one additional degree of freedom, as the MS gates if they are parameterized. Including the global phase, a circuit containing $k$ parameterized MS gates can have at most $(2n + 1) \times k + 3n + 1$ DOF. In the example given figure \ref{MS_circuit} the topology has $24$ DOF. To reach universality we must verify equation (\ref{sard}), which leads to the lower bound
\begin{equation*} \#MS \geq \ceil[\bigg]{\frac{4^n - 3n - 1}{2n+1}}. \end{equation*}

\end{proof}

This proof easily transposes to the state preparation problem with a few changes: 
\begin{itemize}
	\item a quantum state is completely characterized by $2^{n+1} - 2$ real parameters,
	\item starting from the state $\ket{0}^{\otimes n}$ we can only add one degree of freedom to each qubit on the first rotations because the first $R_z$ result in an unnecessary global phase.
\end{itemize}

Consequently the total number of DOF a topology on $n$ qubits with $k$ MS gates can have is at most $(2n+1)k + 2n$. We get the lower bound
\begin{equation*} \#MS \geq \ceil[\bigg]{\frac{2^{n+1} - 2n - 2}{2n+1}}. \end{equation*}

To our knowledge this is the first calculus of a lower bound in the context of trapped-ions circuits. In the
next section we propose an algorithm based on numerical optimization
that achieves accurate circuit synthesis using a number of gates
corresponding to the above lower bounds. This gives a good indication
of the tightness of these bounds.

\section{The optimization framework}\label{framework}

Given a function $f$ representing a topology on $n$ qubits, finding the best possible synthesis of a unitary matrix $U$ with the topology $f$ can be reformulated as solving 
\begin{equation*} \arg \; \underset{x}{\min} \; \| f(x) - U \| := \arg \; \underset{x}{\min} \; g(x), \end{equation*}
where $\| \cdot \|$ is an appropriate norm. We decompose the cost function $g$ as 
\[ g(x) = \frac{1}{k} \sum_{i=1}^{k} \left\|  f(x) \ket{e_{\phi(i)}} - u_{\phi(i)} \right\|,\] 
where $\phi$ is a permutation of $\llbracket 1, 2^n \rrbracket$ and $u_j$ denotes the $j$-th column of $U$. In other words we generalize our problem to the synthesis of $k$ given columns of our operator $U$. For simplicity now we can only study the case where $k=1$: this is the state preparation problem.

Our goal is to rely on various algorithms for non linear optimization.
We choose the norm to be the Euclidean norm
so that with this choice of norm we have a simple expression of the cost error: 

\begin{equation}\label{Eq:error} g(x) = 2 \times \left(1 - \Re\left(\bra{0}f(x)^{\dag}\ket{\psi}\right)\right). \end{equation} 

Hence the cost to compute the error function is equivalent to the simulation cost of the circuit. Starting from the state $\ket{\psi}$ we simulate the circuit $f(x)^{\dag}$, then the real part of the first component gives the error. Many methods for simulating a quantum circuit have been designed and we can rely on them to perform efficient calculations. 

Since $g$ is $C^{\infty}$ we can use more performant optimization algorithms if we can compute the gradient and if possible the Hessian. To compute the gradient, we need to give a more explicit formula for $g(x)$. 
For some $j$ we have
\begin{equation*} \frac{\partial}{\partial x_j} g(x) = -2\frac{\partial}{\partial x_j}\Re\left(\bra{0}f(x)^{\dag}\ket{\psi}\right),\end{equation*}
and by linearity we get 
\begin{equation*} \frac{\partial}{\partial x_j} g(x) = -2\Re\left(\frac{\partial}{\partial x_j}\bra{0}f(x)^{\dag}\ket{\psi}\right).\end{equation*}

Let us suppose we have $K$ gates in our circuit. Then we can write 
\[ f(x)^{\dag} = \prod_{i=1}^K A_i \]
where $(A_i)_{i=1..K}$ is the set of gates on $n$ qubits that compose the circuit $f$. 
In all circuits encountered the parameterized gates are of the form $e^{i\theta \Omega}$ where $\Omega$ can be either $X, Y, Z$ (tensored with the identity if necessary) or more complex hermitians (see the MS gate for instance), $\theta$ is the parameter of the gate. We assume that two gates are not parameterized by the same variable. Let $k_1$ be the index of the gate depending on $x_j$, then we have

\begin{align*} \frac{\partial}{\partial x_j}\bra{0}f(x)^{\dag}\ket{\psi} & = \frac{\partial}{\partial x_j}\bra{0}\prod_{i=1}^{k_1-1} A_i \times e^{i \times x_j \times \Omega} \times \prod_{i=k_1+1}^{K} A_i \ket{\psi} \\ & = \bra{0}\prod_{i=1}^{k_1-1} A_i \times \frac{\partial}{\partial x_j} e^{i \times x_j \times \Omega} \times \prod_{i=k_1+1}^{K} A_i \ket{\psi} \\ & = i \times \bra{0}\prod_{i=1}^{k_1-1} A_i \times \Omega \times A_{k_1} \times \prod_{i=k_1+1}^{K} A_i \ket{\psi} \end{align*} 

Therefore we have a simple expression for the partial derivatives 
\begin{equation*} \frac{\partial}{\partial x_j} g(x) = 2\Im\left( \bra{0}\prod_{i=1}^{k_1-1} A_i \times \Omega \times A_{k_1} \times \prod_{i=k_1+1}^{K} A_i \ket{\psi} \right). \end{equation*}

To compute the whole gradient vector we propose the following algorithm:
\begin{algorithm}
\caption{Computing the gradient of the cost function}
\label{PseudoCode}
\begin{algorithmic} 
\REQUIRE $n > 0, \ket{\psi} \in \mathbb{C}^{2^n}, f : \mathbb{R}^k \to \mathcal{U}(2^n)$, 
\ENSURE Computes $\partial f$
\STATE $\bra{\psi_1} \leftarrow \bra{0}$
\STATE $\ket{\psi_2} \leftarrow f(x)^{\dag} \ket{\psi}$
\STATE $m \leftarrow 1$
\STATE // $N$ is the total number of gates in the circuit
\FOR{$i = 1, N$}
\STATE // $A_i$ is the $i$-th gate in the circuit
\IF{$A_i$ is parameterized}
\STATE // $H_i$ is the Hamiltonian associated to the gate $A_i$
\STATE $df[m] \leftarrow 2\Im\left(\bra{\psi_1} H_i \ket{\psi_2}\right)$
\STATE $m \leftarrow m+1$
\ENDIF
\STATE $\bra{\psi_1} \leftarrow \bra{\psi_1}A_i$
\STATE $\ket{\psi_2} \leftarrow A_i^{\dag} \ket{\psi_2}$
\ENDFOR
\end{algorithmic}
\end{algorithm}

Therefore computing the gradient is equivalent to simulating two times the circuit. On total we need 3 circuit simulations to compute the error and the gradient. 

Note that we need to store two vectors when computing the gradient. If $k > 2^{n-1}$ then we use more memory than if we have only stored the entire matrix.
As memory is not a big concern here compared to the computational time, we keep this framework even for $k=2^n$ for example.

\section{Numerical experiments}\label{experiments}

The experiments have been carried out on one node of the QLM (Quantum Learning Machine) located at ATOS/BULL. This node is a 24-core Intel Xeon(R) E7-8890 v4 processor at 2.4 GHz. Our algorithm is implemented in C with a Python interface and has been executed on 12 cores using OpenMP multithreading.
The uniform random unitary matrices are generated according to the Haar's measure~\cite{doi:10.1137/0717034}. For the numerical optimization we use the BFGS algorithm provided in the SciPy~\cite{scipy} package.

\subsection{Synthesizing generic circuits}

A clear asset of trapped-ion technology is that there is no notion of topology. Contrary to superconducting circuits where we have to deal with the placement of the CNOT gates, we just have to minimize the number of MS gates to optimize the entangling resources. Local rotations are less expensive to realize experimentally~\cite{martinez2016compiling}. So, as a first approach, we can always apply one-qubit gates on every qubits between two MS gates such that a quantum circuit for trapped ions always follow a precise layer decomposition. 

For 50 random unitary matrices on $k \in \{2,3,4\}$ qubits, and
quantum states on $k$ ranging in $\{2,3,4,5,6,7\}$ qubits, we execute Algorithm~\ref{PseudoCode} with circuits containing various numbers of parameterized MS gates. The stopping criterion for the optimization is the one chosen in SciPy i.e., the norm of the gradient must be below $10^{-5}$. We repeat the optimization process several times per matrix with different starting points to maximize our chance to reach a global minimum.
Then we plot, for various numbers of MS gates, 
\begin{itemize}
	\item the error expressed in Formula~(\ref{Eq:error}), maximized over the sample of matrices,
	\item the number of iterations needed for convergence, averaged over the sample of matrices.
\end{itemize}
We summarize our results in Figures \ref{4qubits} and \ref{7qubits} corresponding respectively to the circuit synthesis problem on 4 qubits and the state preparation problem on 7 qubits.
The results obtained for lower number of qubits are similar.
The amount of time required to perform such experiments for larger problems is too important (more than a day).

For both graphs, we observe an exponential decrease of the error with the number of MS gates. This strong decrease shows that there is a range of MS gates count for which we have an acceptable accuracy without being minimal. Although we can save only a small number of MS gates by this trade, we conjecture that for a given precision the number of saved MS gates will increase with the number of qubits. In other words if we want to synthesize an operator up to a given error, the range of ``acceptable'' number of MS gates will increase with the number of qubits.

  \def\setplot#1#2#3{\node [#3] at #2 {#1}; \def\lab{#1}\def\r{#2}\def\pc{#3}}
  \def\plotdata#1{\node [\pc] at #1 {\lab}; \draw[thick, dotted, color=\pc] \r -- #1; \def\r{#1}}

\begin{figure}
  \centering
\subfloat[4-qubits Quantum circuit synthesis problem.\label{4qubits}]{
  \scalebox{.64}{\begin{tikzpicture}[x=2mm,y=.2mm]
    \draw[xstep=5,ystep=44,gray!50,very thin] (-2,-10) grid (40,360);
    \draw[thick,-] (-2,-10) -- node [anchor=north, inner sep=7mm] {Number of MS gates} (40,-10);
    \draw[thick,-] (-2,-10) -- node[anchor=east, inner sep=9mm] {\rotatebox{90}{Synthesis error}} (-2,360);
    \draw[thick,-] (40,-10) -- node[anchor=west, inner sep=9mm] {\rotatebox{90}{\# Iterations}} (40,360);
    \draw[color=red, very thick,-] (27,-10) -- (27,360);
    \foreach \x in {0,5,10,15,20,25,30,35}
    \draw[very thick] (\x,-12) -- (\x,-8) node[anchor=north] {\scalebox{.8}{$\x$}};
    \foreach \x in {0,1,2,3,4,5,6}
    \pgfmathtruncatemacro\result{200*\x}
    \pgfmathtruncatemacro\pos{52*\x-10}
    \draw[very thick] (39.8,\pos) -- (40.2,\pos) node[anchor=west] {\scalebox{.8}{$\result$}};
    \foreach \x in {0,1,2,3,4,5,6,7,8}
    \pgfmathtruncatemacro\y{44*\x}    
    \pgfmathparse{0.2*\x} \let \z \pgfmathresult
    \draw[very thick] (-1.8,\y) -- (-2.2,\y) node[anchor=east] {\scalebox{.8}{$\pgfmathprintnumber[fixed,precision=1]{\z}$}} ;
    \setplot{+}{(0,0)}{blue}
    \plotdata{(1,9)}\plotdata{(3,23)}\plotdata{(6,42)}
    \plotdata{(9,62)}\plotdata{(12,82)}\plotdata{(15,112)}
    \plotdata{(18,146)}\plotdata{(21,196)}\plotdata{(24,265)}\plotdata{(26,336)}
    \plotdata{(27,329)}\plotdata{(28,235)}\plotdata{(29,179)}\plotdata{(30,144)}
    \plotdata{(31,116)}\plotdata{(32,94)}\plotdata{(33,76)}\plotdata{(34,64)}
    \plotdata{(35,54)}\plotdata{(36,46)}\plotdata{(37,42)}\plotdata{(38,38)}
    \setplot{$\star$}{(0,336)}{black}\plotdata{(1,292)}\plotdata{(3,227)}
    \plotdata{(6,146)}\plotdata{(9,91)}\plotdata{(12,50)}\plotdata{(15,25)}
    \plotdata{(18,10)}\plotdata{(21,3)}\plotdata{(24,1)}\plotdata{(26,0)}
    \plotdata{(27,0)}\plotdata{(28,0)}\plotdata{(29,0)}\plotdata{(30,0)}
    \plotdata{(31,0)}\plotdata{(32,0)}\plotdata{(33,0)}\plotdata{(34,0)}
    \plotdata{(35,0)}\plotdata{(36,0)}\plotdata{(37,0)}\plotdata{(38,0)}
    \draw[fill=white] (6,355) rectangle (21,297) ;
    \node at (6,358) [anchor=north west]
    {\scalebox{.7}{\begin{minipage}{4cm}\begin{tabular}{@{}ll@{}}
            \textcolor{blue}{${\cdot}{\cdot}{\cdot}$+${\cdot}{\cdot}{\cdot}$} & Error
            \\[-.2ex]
            \textcolor{black}{${\cdot}{\cdot}{\cdot}{\star}{\cdot}{\cdot}{\cdot}$} & \# Iterations
            \\[-.2ex]
            \textcolor{red}{\raisebox{.5ex}{\rule{5ex}{1.5pt}}} & Lower bound
          \end{tabular}
        \end{minipage}}} ;
  \end{tikzpicture}}
}\hfill
\subfloat[7-qubits State preparation problem.\label{7qubits}]{
\scalebox{.64}{\begin{tikzpicture}[x=2.4118mm,y=.201mm]
    \draw[xstep=5,ystep=60,gray!50,very thin] (-2,-10) grid (32,358);
    \draw[thick,-] (-2,-10) -- node [anchor=north, inner sep=7mm] {Number of MS gates} (32,-10);
    \draw[thick,-] (-2,-10) -- node[anchor=east, inner sep=9mm] {\rotatebox{90}{Synthesis error}} (-2,358);
    \draw[thick,-] (32,-10) -- node[anchor=west, inner sep=9mm] {\rotatebox{90}{\# Iterations}} (32,358);
    \draw[color=red, very thick,-] (16,-10) -- (16,358);
    \foreach \x in {0,5,10,15,20,25,30}
    \draw[very thick] (\x,-12) -- (\x,-8) node[anchor=north] {\scalebox{.8}{$\x$}};
    \foreach \x in {0,1,2,3,4,5,6,7}
    \pgfmathtruncatemacro\result{200*\x}
    \pgfmathtruncatemacro\pos{46.15*\x-4}
    \draw[very thick] (31.8,\pos) -- (32.2,\pos) node[anchor=west] {\scalebox{.8}{$\result$}};
    \foreach \x in {0,1,2,3,4,5}
    \pgfmathtruncatemacro\y{60*\x}    
    \pgfmathparse{0.05*\x} \let \z \pgfmathresult
    \draw[very thick] (-1.8,\y) -- (-2.2,\y) node[anchor=east] {\scalebox{.8}{$\pgfmathprintnumber[fixed,precision=2]{\z}$}} ;
    \setplot{+}{(0,0)}{blue}\plotdata{(2,78)}\plotdata{(4,116)}\plotdata{(6,140)}\plotdata{(8,160)}
\plotdata{(10,178)}\plotdata{(12,207)}\plotdata{(14,267)}\plotdata{(15,308)}\plotdata{(16,336)}\plotdata{(17,202)}
\plotdata{(18,136)}\plotdata{(19,98)}\plotdata{(20,72)}\plotdata{(21,53)}\plotdata{(22,46)}\plotdata{(23,42)}
\plotdata{(24,38)}\plotdata{(25,30)}\plotdata{(26,26)}\plotdata{(27,25)}\plotdata{(28,24)}\plotdata{(29,23)}
    \setplot{$\star$}{(0,338)}{black}
\plotdata{(2,242)}\plotdata{(4,158)}\plotdata{(6,97)}\plotdata{(8,47)}\plotdata{(10,22)}\plotdata{(12,9)}\plotdata{(14,2)}
\plotdata{(15,0)}\plotdata{(16,0)}\plotdata{(17,0)}\plotdata{(18,0)}\plotdata{(19,0)}\plotdata{(20,0)}\plotdata{(21,0)}
\plotdata{(22,0)}\plotdata{(23,0)}\plotdata{(24,0)}\plotdata{(25,0)}\plotdata{(26,0)}\plotdata{(27,0)}\plotdata{(28,0)}\plotdata{(29,0)}
    \draw[fill=white] (18,335) rectangle (31,277) ;
    \node at (18,338) [anchor=north west]
    {\scalebox{.7}{\begin{minipage}{4cm}\begin{tabular}{@{}ll@{}}
            \textcolor{blue}{${\cdot}{\cdot}{\cdot}$+${\cdot}{\cdot}{\cdot}$} & Error
            \\[-.2ex]
            \textcolor{black}{${\cdot}{\cdot}{\cdot}{\star}{\cdot}{\cdot}{\cdot}$} & \# Iterations
            \\[-.2ex]
            \textcolor{red}{\raisebox{.5ex}{\rule{5ex}{1.5pt}}} & Lower bound
          \end{tabular}
        \end{minipage}}} ;
  \end{tikzpicture}}}
\caption{Evolution of the synthesis error/number of iterations with the number of MS gates.}
\label{Error}
\end{figure}
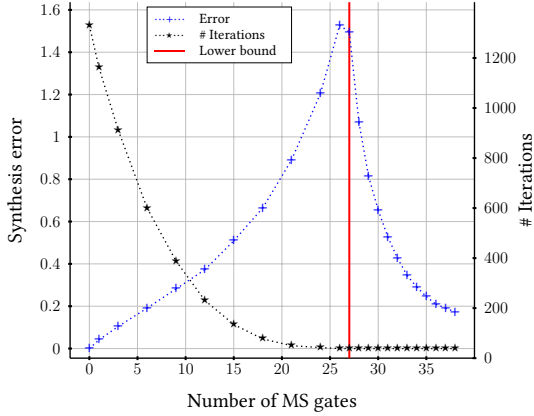
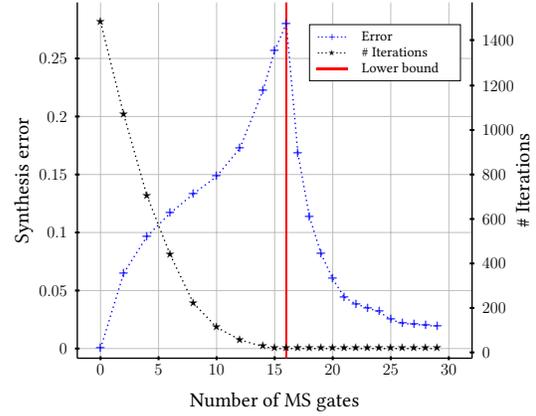

What is the experimental lower bound ? Interestingly, the number of
iterations is a better visual indicator of this lower bound : once the
exact number of DOF is reached, we add redundant degrees of freedom
into the optimization problem. We make the
  hypothesis that this leads to a bigger search space but with many
more global minima. Hence the search is facilitated and the algorithm
can converge in fewer iterations.
So the peak of the number of iterations corresponds to the experimental lower bound, which is confirmed by the computed errors. 

The experimental lower bounds follow the formulas given in Section \ref{lower} except for the 2-qubits quantum circuit synthesis case where we need 1 more MS gate to reach universality. This gives strong indications about the correctness of the theoretical lower bounds and the relevance of the approach of counting the degrees of freedom to estimate the resources necessary for implementing a quantum operator. 

\subsection{Tradeoff quantum/classical cost}

In the previous experiments, we could synthesize unitary operators on 6 qubits in 7 hours and quantum states on 11 qubits in 13 hours, both with circuits of optimal size. 
In the graphs plotted in Figures \ref{4qubits} and \ref{7qubits}, we also observe an exponential decrease of the number of iterations after we have reached the optimal number of MS gates.
We can exploit this behavior if we want to accelerate the time of the synthesis at the price of a bigger circuit. 
By adding only a few more MS gates the number of iterations can be significantly reduced, allowing us to address potentially bigger problems in the synthesis. 
In Figure \ref{Scal} we show, for different problem sizes, how the iteration count and the time per iteration evolve when we add $10\%$ more MS gates to the optimal number of gates.
We observe that the number of iterations is reduced at most by a factor 5 in the case of 6 qubits for generic operators and 11 qubits for quantum states. 
More importantly, with such an augmentation in the number of MS gates, the number of iterations only slightly increases, almost linearly, in contrary to circuits of optimal size where the number of iterations increases exponentially. This means that the time per iteration, also increasing exponentially with the number of qubits, is the main limiting factor in the good scalability of the method. Thus any improvement in the classical simulation of a quantum circuit (i.e., reducing the time per iteration) will lead to significant acceleration in the synthesis time.

Finally in our experiments we achieved circuit synthesis on 6 qubits in about an hour, and a state preparation on 11 qubits in about 3 hours. Despite this sacrifice in the quantum cost (since we use more gates), this is still to our knowledge the best method to synthesize generic quantum circuits and quantum states for the trapped-ions technology.

\begin{figure}
\centering
\makebox[\textwidth][c]
{
\subfloat[Quantum circuit synthesis problem.]{
  \scalebox{.64}{\begin{tikzpicture}[x=1.4cm,y=.2mm]
    \draw[xstep=1,ystep=75,gray!50,very thin] (0,-10) grid (5.5,360);
    \draw[thick,-] (0,-10) -- node [anchor=north, inner sep=7mm] {Number of qubits} (5.5,-10);
    \draw[thick,-] (0,-10) -- node[anchor=east, inner sep=11mm] {\rotatebox{90}{\# Iterations}} (0,360);
    \draw[thick,-] (5.5,-10) -- node[anchor=west, inner sep=11mm]
    {\rotatebox{90}{Time per iteration (s)}} (5.5,360);
    \foreach \x in {2,3,4,5,6}
    \draw[very thick] (\x-1,-12) -- (\x-1,-8) node[anchor=north] {\scalebox{.8}{$\x$}};
    \foreach \x in {0,1,2,3}
    \pgfmathtruncatemacro\result{\x-3}
    \draw[very thick] (5.4,97*\x-5) -- (5.6,97*\x-5) node[anchor=west, inner sep=0.5mm] {\scalebox{.8}{$10^{\result}$}};
    \foreach \x in {0,1,2,3,4}
    \pgfmathtruncatemacro\result{2000*\x}
    \pgfmathtruncatemacro\y{75*\x}    
    \pgfmathparse{0.2*\x} \let \z \pgfmathresult
    \draw[very thick] (0.1,\y) -- (-0.1,\y) node[anchor=east] {\scalebox{.8}{$\result$}} ;
    \setplot{$\bullet$}{(1,2)}{blue}\plotdata{(2,8)}\plotdata{(3,26)}\plotdata{(4,40)}\plotdata{(5,48)} 
    \setplot{$\star$}{(1,2)}{black}\plotdata{(2,11)}\plotdata{(3,49)}\plotdata{(4,129)}\plotdata{(5,304)}
    \setplot{\scalebox{.5}{$\blacksquare$}}{(1,2)}{black}\plotdata{(2,413-387)}\plotdata{(3,413-312)}\plotdata{(4,413-211)}\plotdata{(5,413-83)}
    \setplot{+}{(1,2)}{blue}\plotdata{(2,413-387+1)}\plotdata{(3,413-312+2)}\plotdata{(4,413-211+4)}\plotdata{(5,413-83+6)}
    \draw[fill=white] (0.3,355) rectangle (4.3,280) ;
    \node at (.3,358) [anchor=north west]
    {\scalebox{.7}{\begin{minipage}{10cm}
          \begin{tabular}{@{}ll@{}}
            \textcolor{black}{${\cdot}{\cdot}{\cdot}{\star}{\cdot}{\cdot}{\cdot}$}
            & iter (with min. \# gates)
            \\[-.2ex]
            \textcolor{blue}{${\cdot}{\cdot}{\cdot}{\bullet}{\cdot}{\cdot}{\cdot}$} 
            & \# iter (with min. \# gates + 10\%)
            \\[-.2ex]
            \textcolor{black}{${\cdot}{\cdot}{\cdot}{\scalebox{.5}{$\blacksquare$}}{\cdot}{\cdot}{\cdot}$} 
            & \# time/iter (with min. \# gates)
            \\[-.2ex]
            \textcolor{blue}{${\cdot}{\cdot}{\cdot}${{+}}${\cdot}{\cdot}{\cdot}$} 
            & \# time/iter (with min. \# gates + 10\%)
          \end{tabular}
        \end{minipage}}} ;
  \end{tikzpicture}}
} 
\subfloat[State preparation problem.]{
  \scalebox{.64}{\begin{tikzpicture}[x=7mm,y=.2mm]
    \draw[xstep=1,ystep=56,gray!50,very thin] (0,-10) grid (11.5,360);
    \draw[thick,-] (0,-10) -- node [anchor=north, inner sep=7mm] {Number of qubits} (11.5,-10);
    \draw[thick,-] (0,-10) -- node[anchor=east, inner sep=11mm] {\rotatebox{90}{\# Iterations}} (0,360);
    \draw[thick,-] (11.5,-10) -- node[anchor=west, inner sep=11mm]
    {\rotatebox{90}{Time per iteration (s)}} (11.5,360);
    \foreach \x in {1,2,3,4,5,6,7,8,9,10,11,12}
    \draw[very thick] (\x-1,-12) -- (\x-1,-8) node[anchor=north] {\scalebox{.8}{$\x$}};
    \foreach \x in {0,1,2,3,4}
    \pgfmathtruncatemacro\result{\x-3}
    \draw[very thick] (11.4,72*\x+5) -- (11.6,72*\x+5) node[anchor=west, inner sep=0.5mm] {\scalebox{.8}{$10^{\result}$}};
    \foreach \x in {0,1,2,3,4,5,6}
    \pgfmathtruncatemacro\result{1000*\x}
    \pgfmathtruncatemacro\y{56*\x}    
    \pgfmathparse{0.2*\x} \let \z \pgfmathresult
    \draw[very thick] (0.1,\y) -- (-0.1,\y) node[anchor=east] {\scalebox{.8}{$\result$}} ;
    \setplot{$\bullet$}{(1,2)}{blue}
    \plotdata{(2,412-411)}\plotdata{(3,412-407)}\plotdata{(4,412-395)}\plotdata{(5,412-386)}
    \plotdata{(6,412-378)}\plotdata{(7,412-363)}\plotdata{(8,412-354)}\plotdata{(9,412-368)}\plotdata{(10,412-345)}\plotdata{(11,412-336)}

    \setplot{$\star$}{(1,2)}{black}
    \plotdata{(2,412-407)}\plotdata{(3,412-406)}\plotdata{(4,412-382)}\plotdata{(5,412-368)}
    \plotdata{(6,412-310)}\plotdata{(7,412-333)}\plotdata{(8,412-182)}\plotdata{(9,412-124)}\plotdata{(10,412-77)}

    \setplot{\scalebox{.5}{$\blacksquare$}}{(1,2)}{black}
    \plotdata{(2,413-403)}\plotdata{(3,412-394)}\plotdata{(4,412-369)}\plotdata{(5,412-339)}
    \plotdata{(6,412-301)}\plotdata{(7,412-254)}\plotdata{(8,412-216)}\plotdata{(9,412-168)}\plotdata{(10,412-130)}

    \setplot{+}{(1,2)}{blue}
    \plotdata{(2,412-400)}\plotdata{(3,412-384)}\plotdata{(4,412-363)}\plotdata{(5,412-331)}
    \plotdata{(6,412-297)}\plotdata{(7,412-251)}\plotdata{(8,412-210)}\plotdata{(9,412-158)}\plotdata{(10,412-124)}\plotdata{(11,412-77)}
    \draw[fill=white] (0.3,355) rectangle (8.3,280) ;
    \node at (.3,358) [anchor=north west]
    {\scalebox{.7}{\begin{minipage}{10cm}
          \begin{tabular}{@{}ll@{}}
            \textcolor{black}{${\cdot}{\cdot}{\cdot}{\star}{\cdot}{\cdot}{\cdot}$}
            & iter (with min. \# gates)
            \\[-.2ex]
            \textcolor{blue}{${\cdot}{\cdot}{\cdot}{\bullet}{\cdot}{\cdot}{\cdot}$} 
            & \# iter (with min. \# gates + 10\%)
            \\[-.2ex]
            \textcolor{black}{${\cdot}{\cdot}{\cdot}{\scalebox{.5}{$\blacksquare$}}{\cdot}{\cdot}{\cdot}$} 
            & \# time/iter (with min. \# gates)
            \\[-.2ex]
            \textcolor{blue}{${\cdot}{\cdot}{\cdot}${{+}}${\cdot}{\cdot}{\cdot}$} 
            & \# time/iter (with min. \# gates + 10\%)
          \end{tabular}
        \end{minipage}}} ;
  \end{tikzpicture}}
}
}
\caption{Number of iterations and time per iteration for different problem sizes.}
\label{Scal}
\end{figure}
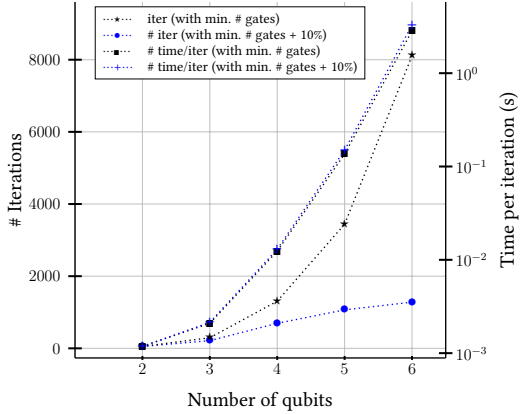
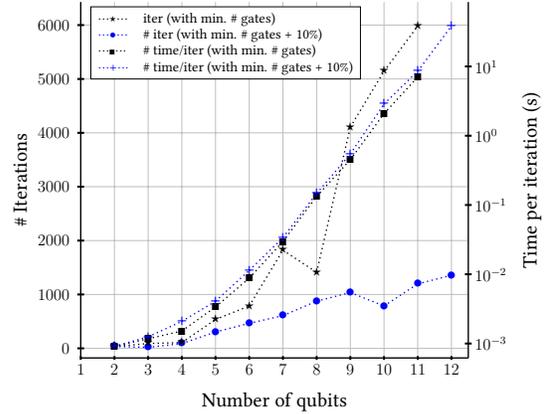

\section{Conclusion}\label{conclusion}

We have explained how a numerical optimization algorithm can be used to synthesize generic quantum circuits adapted to the trapped-ions technology. We have provided a simple algorithm to compute the error and the gradient such that the scalability of our method relies on the simulation cost of the circuits we want to optimize. We have also highlighted a possible tradeoff between the classical time necessary to perform the synthesis and the final size of the circuits resulting in more flexibility in the application of the framework. Finally we have shown that the lower bounds computed experimentally with this framework follow the theoretical results with the hope that it would help for future formal proofs.
 
As future work we plan to extend our analysis to specific quantum operators for which we expect shorter circuits. We also plan to extend this analysis to $\{CNOT, \mathcal{SU}(2)\}$ based circuits, especially in order to answer the following questions : is the lower bound given in \cite{shende2004minimal} a tight one ? Is there a universal topology that reaches that lower bound ? We also investigate a way to efficiently compute the Hessian in order to use more complex optimization methods and we plan to address problems with more qubits thanks to GPU and distributed computing.

\section*{Acknowledgement}
This work was supported in part by the French National Research Agency
(ANR) under the research project SoftQPRO ANR-17-CE25-0009-02,
and by the DGE of the French Ministry of Industry under the research
project PIA-GDN/QuantEx P163746-484124.


\begin{thebibliography}{10}

\bibitem{arrazola2018machine}
J.~M. Arrazola, T.~R. Bromley, J.~Izaac, C.~R. Myers, K.~Bradler, and N.~Killoran.
\newblock Machine learning method for state preparation and gate synthesis on photonic quantum computers.
\newblock {\em Quantum Science and Technology}, 2018.

\bibitem{bravyi2018simulation}
S.~Bravyi, D.~Browne, P.~Calpin, E.~Campbell, D.~Gosset, and M.~Howard.
\newblock Simulation of quantum circuits by low-rank stabilizer decompositions.
\newblock {\em Quantum}, 3:181, 2019.

\bibitem{PhysRevA.87.032332}
B.~Giles and P.~Selinger.
\newblock Exact synthesis of multiqubit clifford+$t$ circuits.
\newblock {\em Physical Review A}, 87:032332, Mar 2013.

\bibitem{Householder}
T.~Goubault~de Brugière, M.~Baboulin, B.~Valiron, and C.~Allouche.
\newblock Quantum circuits synthesis using householder transformations.
\newblock Draft, submitted. 2019.

\bibitem{KRONECKER}
A.~Graham.
\newblock {\em Kronecker products and matrix calculus with application}.
\newblock Wiley, New York, 1981.

\bibitem{Grover:1996:FQM:237814.237866}
L.~K. Grover.
\newblock A fast quantum mechanical algorithm for database search.
\newblock In {\em Proceedings of the Twenty-eighth Annual ACM Symposium on Theory of Computing}, pages 212--219, 1996. ACM.

\bibitem{guillemin2010differential}
V.~Guillemin and A.~Pollack.
\newblock {\em Differential topology}, volume 370.
\newblock 2010.

\bibitem{harrow2009quantum}
A.~W. Harrow, A.~Hassidim, and S.~Lloyd.
\newblock Quantum algorithm for linear systems of equations.
\newblock {\em Physical Review Letters}, 103:150502, 2009.

\bibitem{kerenidis2017recommendation}
I.~Kerenidis and A.~Prakash.
\newblock Quantum recommendation systems.
\newblock In {\em Proceedings of the 8th Conference on Innovations inTheoretical Computer Science Conference}, volume~67 of {\em LIPIcs}, pages
  49:1--49:21. 2017.

\bibitem{kliuchnikov2013fast}
V.~Kliuchnikov, D.~Maslov, and M.~Mosca.
\newblock Fast and efficient exact synthesis of single-qubit unitaries generated by {C}lifford and {T} gates.
\newblock {\em Quantum Information \& Computation}, 13(7-8):607--630, 2013.

\bibitem{knill1995approximation}
E.~Knill.
\newblock Approximation by quantum circuits.
\newblock Technical Report LAUR-95-2225, Los Alamos National Laboratory, 1995.
\newblock Also available as
  \href{https://arxiv.org/abs/quant-ph/9508006}{arXiv:quant-ph/9508006}.

\bibitem{li2016approximate}
R.~Li, U.~Alvarez-Rodriguez, L.~Lamata, and E.~Solano.
\newblock Approximate quantum adders with genetic algorithms: An IBM quantum experience.
\newblock {\em Quantum Measurements and Quantum Metrology}, 4:1--7, 2013.

\bibitem{lukac2003evolutionary}
M.~Lukac, M.~Perkowski, H.~Goi, M.~Pivtoraiko, C.~H. Yu, K.~Chung, H.~Jeech, B.-G. Kim, and Y.-D. Kim.
\newblock Evolutionary approach to quantum and reversible circuits synthesis.
\newblock {\em Artificial Intelligence Review}, 20(3-4):361--417, 2003.

\bibitem{martinez2016compiling}
E.~A. Martinez, T.~Monz, D.~Nigg, P.~Schindler, and R.~Blatt.
\newblock Compiling quantum algorithms for architectures with multi-qubit gates.
\newblock {\em New Journal of Physics}, 18(6):063029, 2016.

\bibitem{mottonen2005decompositions}
M.~Mottonen and J.~J. Vartiainen.
\newblock Decompositions of general quantum gates.
\newblock In S.~Shannon, editor, {\em Trends in Quantum Computing Research}, chapter~7. Nova Science Publishers, New York, 2006.
\newblock Also available online as \href{https://arxiv.org/abs/quant-ph/0504100}{arXiv:quant-ph/0504100}.

\bibitem{nielsen2011quantum}
M.~A. Nielsen and I.~L. Chuang.
\newblock {\em Quantum Computation and Quantum Information}.
\newblock Cambridge University Press, 2011.

\bibitem{PhysRevA.83.032302}
M.~Plesch and i.~c.~v. Brukner.
\newblock Quantum-state preparation with universal gate decompositions.
\newblock {\em Physical Review A}, 83:032302, Mar 2011.

\bibitem{NISQC}
J.~Preskill.
\newblock Quantum {C}omputing in the {NISQ} era and beyond.
\newblock {\em {Quantum}}, 2:79, Aug. 2018.

\bibitem{scipy}
Python.
\newblock {S}cientific {C}omputing {T}ools for {P}ython ({SciPy}).
\newblock \url{https://www.scipy.org}.

\bibitem{1629135}
V.~V. Shende, S.~S. Bullock, and I.~L. Markov.
\newblock Synthesis of quantum-logic circuits.
\newblock {\em IEEE Transactions on Computer-Aided Design of Integrated Circuits and Systems}, 25(6):1000--1010, 2006.

\bibitem{shende2004minimal}
V.~V. Shende, I.~L. Markov, and S.~S. Bullock.
\newblock Minimal universal two-qubit controlled-not-based circuits.
\newblock {\em Physical Review A}, 69(6):062321, 2004.

\bibitem{doi:10.1137/S0036144598347011}
P.~W. Shor.
\newblock Polynomial-time algorithms for prime factorization and discrete
  logarithms on a quantum computer.
\newblock {\em SIAM Review}, 41(2):303--332, 1999.

\bibitem{doi:10.1137/0717034}
G.~Stewart.
\newblock The efficient generation of random orthogonal matrices with an
  application to condition estimators.
\newblock {\em SIAM Journal on Numerical Analysis}, 17(3):403--409, 1980.

\bibitem{PhysRevLett.91.147902}
G.~Vidal.
\newblock Efficient classical simulation of slightly entangled quantum
  computations.
\newblock {\em Physical Review Letters}, 91:147902, Oct 2003.

\bibitem{wright1999numerical}
S.~Wright and J.~Nocedal.
\newblock Numerical optimization.
\newblock {\em Springer Science}, 35(67-68):7, 1999.

\end{thebibliography}

\end{document}